\documentclass[11pt]{article}

\usepackage{fullpage}
\usepackage{url}
\usepackage[small]{caption}
\usepackage{graphicx}
\usepackage{amsmath,amsthm,amssymb}
\usepackage{booktabs}
\usepackage{thm-restate}
\usepackage{mathpazo}
\usepackage[svgnames]{xcolor}
\usepackage{tikz}
\usetikzlibrary{arrows.meta,calc,graphs}
\usepackage{graphicx}
\usepackage{mathtools}
\usepackage{comment}
\usepackage{algorithmic}
\usepackage[mathscr]{eucal}
\usepackage{color} 
\usepackage{cmap}
\usepackage[breaklinks=true]{hyperref}
\hypersetup{colorlinks={true},urlcolor={blue},linkcolor={DarkBlue},citecolor=[named]{DarkGreen}}
\usepackage{cleveref}
\usepackage{nth}
\usepackage{nicefrac}
\usepackage{enumitem}
\usepackage{natbib}
\usepackage{authblk}

\usepackage[ruled,linesnumbered]{algorithm2e} 

\SetAlFnt{\small}
\SetAlCapFnt{\small}
\SetAlCapNameFnt{\small}

\newcommand{\ags}{\mathcal{N}}
\newcommand{\itms}{\mathcal{M}}

\newcommand{\sz}{T}
\newcommand{\cst}{\mathsf{Cost}}

\newcommand{\utp}{\ensuremath{\vec{u}}}
\newcommand{\tauv}{\ensuremath{\vec{\tau}}}
\newcommand{\dist}{\mathsf{dist}}
\newcommand{\sw}{\mathsf{sw}}

\newcommand{\mtch}{A}
\newcommand{\opt}{\mathsf{\mtch^*}}
\newcommand{\osw}{\mathsf{s^*}}

\newcounter{note}[section]

\newtheorem{theorem}{Theorem}[section]
\newtheorem{corollary}[theorem]{Corollary}

\newtheorem{lemma}[theorem]{Lemma}

\theoremstyle{definition}
\newtheorem{definition}[theorem]{Definition}

\newtheorem{example}[theorem]{Example}
\theoremstyle{remark}
\newtheorem{remark}[theorem]{Remark}

\DeclarePairedDelimiter{\set}{\{}{\}}

\DeclareMathOperator{\E}{\mathbb{E}}

\title{\bf The Distortion of Threshold Approval Matching}

\author[1]{Mohamad Latifian}
\affil[1]{Department of Computer Science, University of Toronto, Canada}
\author[2]{Alexandros A. Voudouris}
\affil[2]{School of Computer Science and Electronic Engineering, University of Essex, UK}

\date{}

\begin{document} 

\allowdisplaybreaks

\maketitle

\begin{abstract}
We study matching settings in which a set of agents have private utilities over a set of items. Each agent reports a partition of the items into approvals sets of different threshold utility levels. Given this limited information on input, the goal is to compute an assignment of the items to the agents (subject to cardinality constraints depending on the application) that (approximately) maximizes the social welfare (the total utility of the agents for their assigned items). We first consider the well-known, simple one-sided matching problem in which each of $n$ agents is to be assigned exactly one of $n$ items. We show that with $t$ threshold utility levels, the distortion of deterministic matching algorithms is $\Theta(\sqrt[t]{n})$ while that of randomized algorithms is $\Theta(\sqrt[t+1]{n})$. We then show that our distortion bounds extend to a more general setting in which there are multiple copies of the items, each agent can be assigned a number of items (even copies of the same one) up to a capacity, and the utility of an agent for an item depends on the number of its copies that the agent is given.  
\end{abstract}

\section{Introduction}\label{sec:intro}
The assignment of papers to reviewers in conference management systems like CMT, HotCRP, EasyChair and OpenReview is computed using bidding information that classifies the papers into sets based on whether the reviewers are, for example, eager, willing, or not willing to handle them. In a sense, this process defines a collection of {\em threshold levels} that the reviewers (or, more generally, {\em agents}) can use to partition the papers (or, more generally, {\em items}) into associated {\em approval sets} based on their preferences (which can be dependent on their experience, their interests, and so on).  

Eliciting only threshold approvals rather than more detailed information about the underlying utility preferences of the agents for the items inevitably leads to inefficiency in terms of natural, cardinal objectives such as the well-known {\em social welfare} (the total utility). Typically, the loss of efficiency of decision-making methods that have access only to incomplete information is captured by the notion of {\em distortion}, which is defined as the worst-case ratio of the maximum possible social welfare over that of the computed solution. The distortion was originally used for social choice settings (such as voting) where decisions are made only based on ordinal information (rankings) \citep{procaccia2006distortion,boutilier2015optimal}, but has recently been studied for settings in which different types of information is available or can be elicited (e.g., see \citep{amanatidis2021peeking,amanatidis2021matching,mandal2019efficient,mandal2020optimal,ma2021matching}). 

In the matching setting we study in this work, which captures various interesting applications (such as the {\em paper assignment problem} in peer-reviewing that we briefly introduced above, as well as general {\em constrained resource allocation}), the threshold approvals reported by the agents is a type of information that lies in-between fully cardinal and fully ordinal. Hence, while we cannot hope to achieve full efficiency, we can hope to achieve distortion better than what is possible just with ordinal preferences, depending on how detailed the threshold approvals are. In particular, we are interested in the possible {\em tradeoffs} between the distortion and the number of threshold levels both for when allocations are computed deterministically (which is the most natural way of doing so in social choice problems), as well as when randomization can be exploited.

\subsection{Our Contribution}
We start by considering the fundamental {\em one-sided matching problem} (also known as {\em house allocation}) to introduce the main ideas of our techniques, before turning to a more general setting. In one-sided matching, there is a set of $n$ agents with utilities for a set of $n$ items; we assume that the utilities satisfy the standard unit-sum assumption \citep{aziz19justifications}. The utilities are private and are not explicitly reported by the agents. Instead, for a number $t$ of decreasing threshold values, each agent reports a collection of $t$ {\em approval sets} consisting of items of different utility level; in particular, each approval set is associated with a threshold value and includes all items for which the agent has utility that is at least this threshold. Given the approval sets as input, our goal is to determine a one-to-one matching between agents and items so that the social welfare (total utility of the agents for their assigned items) is maximized. 

We show tight bounds on the best possible distortion achieved by matching mechanisms for any number $t$ of thresholds. In particular, we show a bound of $\Theta(\sqrt[t]{n})$ for deterministic mechanisms, and a bound of $\Theta(\sqrt[t+1]{n})$ for randomized ones. The lower bounds are presented in Section~\ref{sec:lower} and the upper bounds in Section~\ref{sec:upper}. To put the bounds into perspective, we note that just one threshold is sufficient to obtain distortion $\Theta(n)$ for deterministic algorithms, beating the best possible distortion of $\Theta(n^2)$ that can be achieved using ordinal information~\citep{amanatidis2021matching}. Similarly, a distortion of $\Theta(\sqrt{n})$ can be achieved with randomization, matching the best possible distortion achieved by ordinal randomized algorithms~\citep{filos2014RP}.


In Section~\ref{sec:extention}, we turn our attention to a more general setting where the agents have {\em capacities}, indicating the maximum number of items they can receive, and the items have {\em supplies}, indicating the number of copies of them that are available. We make a budget-balance assumption that the total capacity is asymptotically of the same order of the total supply; for example, in the paper assignment problem, all papers must receive a number of reviews, and so the total capacity must be sufficiently larger than the total supply. We also assume that, when agents can receive multiple copies of an item, their utility depends on the number of these copies, and thus the copies are not treated as independent items. Our goal is to compute an allocation of items to agents, such that the capacity and the supply constraints are satisfied, and the social welfare is maximized. 

As this setting is a generalization of the one-sided matching (in which the number of agents is equal to the number of items, and there are unit capacities and supplies), our lower bounds from Section~\ref{sec:lower} extend directly. For the upper bounds, we show that, for $t$ thresholds, the best possible distortion achieved by deterministic mechanisms is $O(c \cdot \sqrt[t]{\sz})$, while the best distortion achieved by randomized mechanisms is $O(c \cdot \sqrt[t+1]{\sz})$, where $\sz$ is the total available supply (or capacity) and $c$ is a parameter that depends either on the maximum capacity or the ratio between the number of items and agents. From this, we get bounds $\Theta(\sqrt[t]{n})$ and $\Theta(\sqrt[t+1]{n})$ when the capacities and the supplies are constant.

\subsection{Related Work}
The distortion was originally defined by \citet{procaccia2006distortion} to measure the worst-case loss in social welfare when voting decisions are made using only ordinal information. Since then, the distortion has been studied for several different voting problems, including utilitarian voting~\citep{procaccia2006distortion,boutilier2015optimal,Caragiannis2011embedding,caragiannis2017subset,ebadian2022optimized}, 
metric voting~\citep{anshelevich2018approximating,CSV22,charikar2022randomized,charikar24breaking,gkatzelis2020resolving,pluralityveto2022}, and combinations of the two~\citep{gkatzelis2023both}. It has also been studied for social choice problems beyond voting, such as one-sided matching that we also consider in this paper~\citep{filos2014RP,amanatidis2021matching,amanatidis2022dice}, as well as other clustering and graph problems~\citep{abramowitz2017utilitarians,anshelevich2016blind,caragiannis2024low}. See the survey of \citet{survey2021} for an introduction to the distortion framework and the various models that have been considered. 

Our paper follows a relatively recent stream of papers within the distortion literature that have considered elicitation methods beyond ordinal information. In this direction, \citet{mandal2019efficient,mandal2020optimal} showed tradeoffs between the best possible distortion and a communication complexity measure (the number of bits the agents can use to report information) for utilitarian voting. Results of similar flavour for metric voting have also been showed, for example, by \citet{kempe2020communication}.

More related to our work, in a series of papers, \citet{amanatidis2021peeking,amanatidis2021matching,amanatidis2022dice} studied voting and matching settings in which the agents provide ordinal information and, on top of that, are capable of answering {\em value queries} about their utilities for specific alternatives. They showed lower and upper bounds on the distortion of deterministic mechanisms that are functions of the number of queries per agent that are of similar to ours in the sense that the distortion decreases with the number of queries; some of their lower bounds (related to the number of queries required to achieve constant distortion) were recently improved by \citep{caragiannis2023impartial}. Another related paper is that of \citet{ma2021matching} who considered the one-sided matching problem when the agents can answer binary threshold queries about whether their utility for specific alternatives is larger than appropriately chosen thresholds. Using an approach similar to that of \citeauthor{amanatidis2021matching}, they showed bounds on the distortion of deterministic mechanisms that is a function of the number of queries in terms of the social welfare among matchings that satisfy properties such as Pareto or rank optimality. Ignoring differences in the models, the elicitation methods in these papers are related to the one we consider since the threshold approval sets can be computed using a number of (value or binary threshold) queries. Hence, our elicitation method is in a sense a bit more demanding. However, for the setting we focus here, we are able to show asymptotically tight bounds not only for deterministic mechanisms, but also for randomized ones, which have not been studied before.  

Threshold approvals have also been recently explored in various other works on the distortion of voting mechanisms, most notably by \citet{ebadian2023approval} who showed that a single, appropriately chosen threshold is sufficient to achieve a distortion of $O(\sqrt{m})$ in utilitarian single-winner voting with $m$ alternatives. In metric voting, \citet{anshelevich2024threshold} showed improved distortion bounds using an approval set per agent computed by a threshold value that is relative (rather than absolute) to the distance from the top-ranked alternative. Threshold approvals have also been considered in the context of participatory budgeting by \citet{benade2021participatory}, and voting under truthfulness constraints by \citet{bhaskar2018truthful}. All these works use just a single threshold, whereas we here explore the full potential of this elicitation method (for matching problems, rather than voting) using multiple thresholds and show tight bounds on the possible distortion.

\section{The One-Sided Matching Problem}\label{sec:prelim}
We start with the simple {\em one-sided matching} setting to express the core idea; in Section~\ref{sec:extention}, we show that our results extend to a more general setting that more accurately captures applications such as paper assignment in peer reviewing. Let $\ags$ be a set of $n$ \textit{agents} and $\itms$ be a set of $n$ \textit{items}. Agent $i$ has a \textit{utility function} $u_i: \itms \to [0, 1]$ over the items. We assume that these utility functions satisfy the unit-sum assumption, which means for each agent $i \in \ags$, $\sum_{j \in \itms} u_i(j) = 1$. Together, these utility functions form the \textit{utility profile }$\utp$. A \textit{matching} of the items to the agents is a bijection $\mtch:  \ags \to \itms$. With a slight abuse of notation, we use $A_i = \mtch(i)$ to refer to the item matched to agent $i$, and also $A(a)$ to refer to the agent matched to item $a$. We define the \textit{social welfare} of a matching $\mtch$ under utility profile $\utp$ to be the total utility of the agents for the items they are matched to, i.e., 
$$\sw(\mtch, \utp) = \sum_{i \in \ags} u_i(\mtch_i) = \sum_{a \in \itms} u_{A(a)}(a).$$ 
The goal is to compute a matching with high social welfare in the worst case.
For ease of notation, we will drop $\utp$ from $\sw(\mtch,\utp)$ whenever it is clear from context.

\paragraph{Elicitation method.}
In this paper, we focus on eliciting threshold approval votes. A \textit{threshold vector} $\tauv = (\tau_1, \ldots, \tau_t)$, we ask each agent $i$ to submit $t$ disjoint {\em threshold approval subsets} of $\itms$, denoted by $S_{i,1}, \ldots, S_{i,t}$, where $S_{i,k}$ includes the items for which the agent has utility in [$\tau_{k-1}, \tau_k)$, with $\tau_0 := 1$. In other words, $S_{i, k} = \{j \in \itms \colon \tau_{k-1} \geq  u_i(j) > \tau_k\}$. All these $n \times t$ threshold approval sets form the \textit{input profile} $\mathbf{S}$. Note that different utility profiles might induce the same input profile. We say that a utility profile $\utp$ is consistent with an input profile $\mathbf{S}$ (in which say we write $\utp \rhd \mathbf{S}$) if for each agent $i\in \ags$, $k \in [t]$, and $j \in S_{i, k}$, $\tau_{k-1}\geq  u_i(j) > \tau_k$.

\paragraph{Mechanisms and distortion.}
A \textit{mechanism} $f$ defines a threshold vector $\tauv$, takes an input profile $\mathbf{S}$ based on $\tauv$, and then outputs a matching $f(\mathbf{S})$ of the items to the agents. The distortion of a matching $\mtch$ on input profile $\mathbf{S}$ is defined as:
$$\dist(\mtch, \mathbf{S}) =  \max_{\utp \rhd \mathbf{S}} \frac{\sw(\opt, \utp)}{\sw(\mtch, \utp)},$$ where $\opt$ is the matching with the maximum social welfare with respect to $\utp$. The distortion of a matching mechanism $f$ is defined as the worst case distortion of $f$ on any input profile:
$$\dist(f) = \sup_{\mathbf{S}} \dist(f(\mathbf{S}), \mathbf{S}).$$

\section{Lower bounds} \label{sec:lower}
In this section we show lower bounds on the best possible distortion achievable by deterministic and randomized mechanisms for the one-sided matching problem. In particular, for mechanisms that use $t \geq 1$ thresholds, we show a lower bound of $\Omega(\sqrt[t]{n})$ for deterministic mechanisms and a lower bound of $\Omega(\sqrt[t+1]{n})$ for randomized mechanisms. We start by showing a technical lemma that holds for randomized mechanisms that will be useful in establishing the lower bounds in several cases. For a randomized mechanism $f$, denote by $p(i,a)$ the probability that item $a$ is assigned to agent $i$ according to $f$.

\begin{restatable}{lemma}{sumProb}
\label{lem:sum-probabilities}
For any subset of items $M \subseteq \itms$, let $A_M$ be the matching of the items in $M$ to the agents with minimum sum of probabilities with respect to $f$. 
Then, $\sum_{a \in M} p(A_M(a),a) \leq 1.$
\end{restatable}

\begin{proof}
We will prove the claim by constructing a matching $A$ with sum of probabilities at most $1$; then, since $A_M$ is the matching with minimum sum of probabilities, the same must hold for it as well. 
Consider the following greedy algorithm: Starting with $T = M$ and $N = \ags$, in each iteration, find the pair $(i,a) \in N \times T$ with the minimum possible $p(i,a)$, assign $a$ to $i$, and then remove $a$ and $i$ from $T$ and $N$, respectively. 

Let $P_j$ be the sum of probabilities of the assigned items at the end of the $j$-th iteration. We will show by induction that $P_j \leq j/|M|$. For $j=1$, by the pigeonhole principle, we can find a pair with $p(i, a) \leq 1/|M|$. So, we now assume that our hypothesis holds for any $j < \ell$, and want to prove that $P_\ell \leq \ell / |M|$. Let $(i, a)$ be a pair with minimum $p(i,a)$ in the $\ell$-th iteration. 
Since the algorithm is greedy, we know that for any item $b$ that was previously assigned to agent $A(b)$, $p(A(b), b) \leq p(i, b)$. 
This means that the sum of the probabilities of $i$ receiving any item assigned in a previous iteration is at least $P_{\ell-1}$. 
Hence, by the pigeonhole principle again,
$p(i, a) \leq \frac{1-P_{\ell-1}}{|M|-\ell+1} $, which implies that
\begin{align*}
P_\ell = P_{\ell-1} + p_{i, a} 
&\leq P_{\ell-1} + \frac{1-P_{\ell-1}}{|M|-\ell+1} \\
&\leq \frac{\ell-1}{|M|}+\frac{1-\frac{\ell-1}{|M|}}{|M|-\ell+1} = \frac{\ell-1}{|M|}+\frac{\frac{|M|-\ell+1}{|M|}}{|M|-\ell+1} = \frac{\ell}{|M|}.
\end{align*}
Consequently, we overall have that $\sum_{a \in M} p(A_M(a),a) \leq P_{|M|} \leq 1$.
\end{proof}

We are now ready to show the lower bounds via a sequence of lemmas capturing different cases. The first lower bound depends on the ratio of consecutive threshold levels and holds for any mechanism (randomized or deterministic). 

\begin{restatable}{lemma}{thresholdGap}
\label{lem:threshold_gap}
Consider a threshold vector $\tauv = (\tau_1, \ldots, \tau_t)$, and let $k \in [t]$ be such that $\delta = \tau_{k-1}/\tau_{k}$ is the largest multiplicative gap between two consecutive thresholds (assuming $\tau_0 = 1$).
Then, the distortion of any matching mechanism $f$ that uses $\tauv$ is $\Omega(\delta)$.
\end{restatable}
\begin{proof}
For constant $\delta$ the bound is trivial, so consider $\delta > 2$. For an infinitesimal $\varepsilon$, define $\tau_k^+ := \tau_k + \varepsilon$, and find integer $m$ such that $\tau_{k-1}/2 + m\tau_k^+ \leq 1 < \tau_{k-1}/2 + (m+1)\tau_k^+$. We partition the items into three sets: $\itms_1$ with $m+1$ items, $\itms_2$ with $1$ item, and $\itms_3$ with $n-m-2$ items. 
Note that since $\varepsilon$ is an infinitesimal, we have 
\begin{equation}
    \label{eq:m_bound}
    \tau_{k-1}/2 + (m+1)\tau_k^+  > 1 \implies (m+1)\tau_k^+ > 1/2 \implies m+1 > \frac{\delta}{2}.
\end{equation}

Now consider an instance with input profile such that 
\begin{itemize}
    \item $S_{i,k} = \itms_1$,
    \item $S_{i,k'} = \itms_2$ for some $k' > k$, and
    \item $S_{i,j} = \emptyset$ for $j \in [t] \setminus \set{k, k'}$
\end{itemize}
We define the following consistent utility profile (that is, the utilities of each agent sum up to $1$ and the utilities induce the aforementioned input profile). 
Let $A_{\itms_1}$ be the matching of the items in $\itms_1$ to the agents with minimum sum of probabilities, and let $\ags_1$ be the set of agents that are assigned an item from $\itms_1$ according to the matching $A_{\itms_1}$, that is, $\ags_1:= \set{i \in \ags: \exists a \in \itms_1 , A_{\itms_1}(a) = i}$.
For fixed item $a^* \in \itms_1$, we define the utility function
$$u_i(a) = \begin{cases}
    \tau_{k-1}/2 & a \in \itms_1, A_{\itms_1}(a) = i \\
    \tau_{k-1}/2 & a = a^*, i \notin \ags_1 \\
    1-\tau_{k-1}/2-m\tau_k^+ & a \in \itms_2 \\
    0 & a \in \itms_3\\
    \tau_{k}^+ & o.w.
\end{cases}$$
Observe that, for any agent $i$, the utilities sum up to $1$. For any agent $i$ there is exactly one item in $\itms_1$ for which $i$ has value $\tau_{k-1}/2$ (either $i \in \ags_1$ and thus there is a single item $a \in \itms_1$ such that $A_{\itms_1}(a)=i$ for which $i$ has value $\tau_{k-1}/2$, or $i \not\in \ags_1$ and thus $i$ has value $\tau_{k-1}/2$ for $a^*$), and $|\itms_1|-1 = m$ items for which $i$ has value $\tau_k^+$. Consequently, 
\begin{align*}
\sum_{a \in \itms} u_i(a) = \frac{\tau_{k-1}}{2} + m \tau_k^+ + 1- \frac{\tau_{k-1}}{2}-m\tau_k^+ = 1.
\end{align*}
Also, observe that the utility function is consistent to the input profile: 
for any item $a \in M_1$, the utility of agent $i$ is either $\tau_{k-1}/2$ or $\tau_k^+$, i.e., in the interval $[\tau_{k-1},\tau_k)$;
for the single item $a \in M_2$, by the choice of $m$, the utility of agent $i$ is 
$$1-\frac{\tau_{k-1}}{2}-m\tau_k^+ <  \frac{\tau_{k-1}}{2}+(m+1)\tau_k^+ -\frac{\tau_{k-1}}{2}-m\tau_k^+ = \tau_k^+ = \tau_k + \varepsilon.$$ Since the inequality is strictly, the utility is at most $\tau_k$ for infinitesimal $\varepsilon$.

Now consider the expected social welfare of $f(S)$ when $\varepsilon \to 0$ according to the above utility profile. The maximum achieved utility is $\tau_{k-1}/2$ from $a^*$, 
$$p(A_{\itms_1}(a), a)\cdot \tau_{k-1}/2 +  (1-p(A_{\itms_1}(a), a))\cdot\tau_k$$ 
from each item $a\in \itms_1 \setminus \{a^*\}$, and at most $\tau_k$ from the item in $\itms_2$. 
By Lemma~\ref{lem:sum-probabilities}, we have that $\sum_{a \in \itms_1} p(A_{\itms_1}(a), a) \leq 1$, and thus
\begin{align*}
\E_{\mtch \sim f(S)}[\sw(\mtch, \utp)] 
&\leq \frac{\tau_{k-1}}{2} + \sum_{a \in \itms_1} \bigg( p(A_{\itms_1}(a), a) \cdot \left(\frac{\tau_{k-1}}{2}-\tau_k \right) + \tau_k \bigg) + \tau_k \\
&\leq \tau_{k-1} + (m+1) \tau_k.
\end{align*}
In the optimal matching, we can assign each item $a \in \itms_1$ to agent $A_{\itms_1}(a)$ for a social welfare of $(m+1) \tau_{k-1}/2$. Due to \eqref{eq:m_bound}, the distortion is
\begin{align*}
\frac12 \cdot \frac{(m+1)\tau_{k-1}}{\tau_{k-1} + (m+1)\tau_k} 
&\geq \min\left(\frac{m+1}{2}, \frac{\tau_{k-1}}{2\tau_k}\right) \in \Omega(\delta).
\end{align*}
and the proof is complete. 
\end{proof}

Our next two lemmas provide lower bounds for deterministic and randomized mechanisms, respectively, for when the last threshold level is sufficiently small. 

\begin{lemma}
\label{lem:deterministic:last_threshold}
Consider a threshold vector $\tauv = (\tau_1, \ldots, \tau_t)$ such that $\tau_t \geq 1/(n-1)$. Then, the distortion of any deterministic matching mechanism $f$ that uses $\tauv$ is unbounded.
\end{lemma}

\begin{proof}
Consider the input profile $\mathbf{S}$ where the threshold approval sets of any agent are empty, and thus the utility of any agent for any item is at most $\tau_t$. Let $A=f(\mathbf{S})$ be the matching computed by the deterministic matching mechanism $f$, and let $B$ be another matching such that $A(a) \neq B(a)$ for every item $a \in \itms$. Consider the utility profile $\utp$ where agents have utility $0$ for their matched item in $A$, utility $\tau_t$ for their matched item in $B$, and $(1-\tau_t)/(n-2)$ for each of the remaining $n-2$ items. Note that $\tau_t \ge 1/(n-1) \implies (1-\tau_t)/(n-2) \leq \tau_t$, and hence $\utp \rhd \mathbf{S}$. 
Since $\sw(A,\utp) = 0$ and $\sw(B,\utp) = n \cdot \tau_t > 0$, the distortion is unbounded. 
\end{proof}

\begin{lemma}
\label{lem:randomized:last_threshold}
Consider a threshold vector $\tauv = (\tau_1, \ldots, \tau_t)$ such that $\tau_t > 1/n$.
Then, the distortion of any randomized matching mechanism $f$ that uses $\tauv$ is $\Omega(n\cdot \tau_t)$.
\end{lemma}

\begin{proof}
Consider the input profile $\mathbf{S}$ where the threshold approval sets of any agent are empty, and thus the utility of any agent for any item is at most $\tau_t$. 
Let $A_\itms$ be the matching over $\itms$ with minimum sum of probabilities; by Lemma~\ref{lem:sum-probabilities},  $\sum_{a \in \itms} p(A_\itms(a), a) \leq 1$. 
Now, consider the utility profile $\utp$ where each agent has utility $\tau_t$ for the item she is matched to according to $A_{\itms}$ and utility $(1-\tau_t)/(n-1)$ for each of the remaining items. Note that $\tau_t \ge 1/n \implies (1-\tau_t)/(n-1) \leq \tau_t$, and hence $\utp \rhd \mathbf{S}$. The expected social welfare of the mechanism is
\begin{align*}
\E_{A \sim f(S)}\left[\sum_{a \in \itms} u_{A(a)}(a)\right]
=& \sum_{a \in \itms} \bigg( p(A_\itms(a), a) \cdot \tau_t + \\&(1- p(A_\itms(a), a) \cdot \frac{1-\tau_t}{n-1}  \bigg) \\
= \left( \tau_t - \frac{1-\tau_t}{n-1} \right) \sum_{a \in \itms}& p(A_\itms(a), a)  + n \cdot \frac{(1-\tau_t)}{n-1} \\
\leq \tau_t - \frac{1-\tau_t}{n-1} + n \cdot &\frac{(1-\tau_t)}{n-1} = 1.
\end{align*}
Since $\sw(A_{\itms},\utp) = n\cdot \tau_t$, the distortion is at least this much. 
\end{proof}

By appropriately combining Lemmas~\ref{lem:threshold_gap}, \ref{lem:deterministic:last_threshold}, and \ref{lem:randomized:last_threshold}, we can establish the desired lower bounds on the distortion of the different types of mechanisms. 

\begin{theorem}
The distortion of any deterministic matching mechanism $f$ that uses a threshold vector $\tauv$ of length $t$ is $\Omega(\sqrt[t]{n})$. 
\end{theorem}

\begin{proof}
If $\tau_t \geq 1/(n-1)$, by Lemma~\ref{lem:deterministic:last_threshold}, the distortion is unbounded.
Otherwise, if $\tau_t \leq 1/(n-1)$, then there exists $k \in [t]$ such that $\delta \geq \tau_{k-1}/\tau_k \geq \tau_t^{-t} \geq \sqrt[t]{n}$, and thus, by Lemma~\ref{lem:threshold_gap}, the distortion is $\Omega(\sqrt[t]{n}).$
\end{proof}

\begin{theorem}
The distortion of any matching mechanism $f$ that uses a threshold vector $\tauv$ of length $t$ is $\Omega(\sqrt[t+1]{n})$. 
\end{theorem}

\begin{proof}
Suppose that the threshold vector $\tauv$ is such that $\tau_t > n^{-t/(t+1)}$. Since $n^{-t/(t+1)} \geq n^{-1}$, by Lemma~\ref{lem:randomized:last_threshold}, the distortion of $f$ is $\Omega(n\cdot \tau_t) = \Omega(\sqrt[t+1]{n})$. 
So, we can now assume that $\tau_t \leq n^{-t/(t+1)}$ and let $k \in \arg\max_{j \in [t]} \tau_{j-1}/\tau_j$ with $\tau_0 = 1$.  
Clearly,  
\begin{align*}
\bigg( \frac{\tau_{k-1}}{\tau_k} \bigg)^t \geq \prod_{j \in [t]} \frac{\tau_{j-1}}{\tau_j} = \frac{1}{\tau_t}
\implies \delta = \frac{\tau_{k-1}}{\tau_k} \geq \tau_t^{-t}.
\end{align*}
By Lemma~\ref{lem:threshold_gap}, the distortion of $f$ is $\Omega(\delta) = \Omega(\tau_t^{-t}) = \Omega(\sqrt[t+1]{n})$. 
\end{proof}

\section{Upper Bounds} \label{sec:upper}
In this section we present asymptotically tight upper bounds for deterministic and randomized matching mechanisms. 
Our deterministic mechanism (described below) computes a maximum-weight matching by assuming that each agent has the minimum possible utility (according to the thresholds) for all the items in the different approval set given as input. 

\begin{definition}
For $\delta > 1$ and $t \in [n]$, consider the threshold vector $\tauv = (\delta^{-1}, \delta^{-2}, \ldots, \delta^{-t})$. 
The \textit{deterministic matching mechanism $f_t$} uses the threshold vector $\tauv$ and, given an input profile $\mathbf{S}$, 
constructs the following weighted bipartite graph $G_\mathbf{S}$: 
There are $2n$ nodes in total, consisting of a node $v_i$ for each agent $i \in \ags$ on the left side and a node $z_a$ for each item $a \in \itms$ on the right side. 
For $i \in \ags$, $k \in [t]$ and $a \in S_{i,k}$, there is an edge from $v_i$ to $z_a$ with weight $w(v_i,z_a)$. 
The mechanism $f_t$ finds the maximum weighted matching in $G_\mathbf{S}$ and, for each matched pair $(v_i, z_a)$, assigns item $a$ to agent $i$. If there are unmatched pairs remaining, $f_t$ completes the allocation arbitrarily. 
\end{definition}
    

\begin{example}
\label{exm:matching}
Let $t=2$ and $\tauv = (\tau_1, \tau_2)$. Suppose that $S_{1,1} = \set{a,c}$, $S_{2,1}=\set{d}, S_{2,2}=\set{c}, S_{3,2}=\set{a, c, d}$, while the remaining approval sets are empty. Mechanism $f_t$ constructs the graph $G_\mathbf{S}$ shown in Figure~\ref{fig:matching}, computes a maximum-weight matching, and then assigns any unmatched items arbitrarily.
\end{example}

\begin{figure}[ht]
    \centering
\begin{tikzpicture}[scale=0.95]
\begin{scope}[every node/.style={inner sep=0,outer sep=1mm,fill, circle, minimum size=2mm},every label/.style={rectangle,fill=none}]
    \node [label={left:$v_1$}] (V1) at (0,0) {};
    \node [label={left:$v_2$}](V2) at (0,-1.5) {};
    \node [label={left:$v_3$}](V3) at (0,-3) {};
    \node [label={left:$v_4$}](V4) at (0,-4.5) {};
    \node [label={right:$z_a$}] (Ua) at (8,0) {};
    \node [label={right:$z_b$}](Ub) at (8,-1.5) {};
    \node [label={right:$z_c$}](Uc) at (8,-3) {};
    \node [label={right:$z_d$}](Ud) at (8,-4.5) {};
\end{scope}

\draw (V1) to node[above,sloped] {$\tau_1$}  (Ua);
\draw (V1) to node[above,sloped,pos=0.3] {$\tau_1$}  (Uc);
\draw (V2) to node[above,sloped,pos=0.75] {$\tau_1$}  (Ud);
\draw (V2) to node[above,sloped] {$\tau_2$}  (Uc);
\draw (V3) to node[below,sloped, pos=0.8] {$\tau_2$}  (Ua);
\draw (V3) to node[above,pos=0.2] {$\tau_2$}  (Uc);
\draw (V3) to node[below,sloped] {$\tau_2$}  (Ud);

\end{tikzpicture}
  
    \caption{The graph $G_\mathbf{S}$ that is used by $f_t$ in Example~\ref{exm:matching}.}
    \label{fig:matching}
\end{figure}
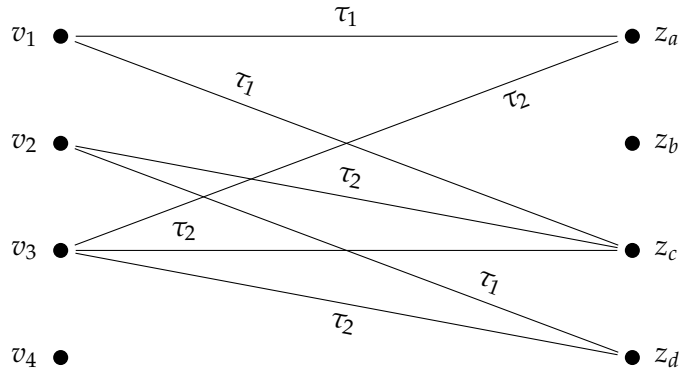

Before we bound the distortion of the mechanism, we prove two very useful technical lemmas. The first one provides us with a lower bound on the weight of the maximum-weight matching in a bipartite graph whose nodes satisfy certain properties; this will be used extensively to lower bound the social welfare of the matching computed by $f_t$, and also by the deterministic mechanism in Section~\ref{sec:extention}.

\begin{restatable}{lemma}{biGraph}
\label{lem:bigraph}
Consider a weighted bipartite graph $G$ with $n$ nodes on the left side $\{v_1, \ldots, v_n\}$ and $m$ nodes $\{z_1, \ldots, z_m\}$ on the right side. If $\sum_{a \in [m]} w(v_i, z_a) \geq W$ for each $v_i$, and $w(v_i, z_a) \geq L$ for each edge $(v_i, z_a)$, then there is matching in $G$ with weight at least $\min\{W, nL\}$.
\end{restatable}

\begin{proof}
We construct a matching $B$ in $G$ as follows: For every $a \in [m]$, we match node $z_a$ to a node $v_i$ such that the edge $(v_i,z_a)$ exists in $G$ and $w(v_i,z_a)$ is the maximum possible among all available (not already matched) nodes in $\{v_1, \ldots, v_n\}$. 
If all nodes in $\{v_1, \ldots, v_n\}$ are matched to a node in $\{z_1, \ldots, z_m\}$, then since the weight of any edge is at least $L$, the total weight of $B$ is at least $nL$. 
Otherwise, if there exists node $v_i$ that is not matched to any node in $\{z_1, \ldots, z_m\}$ according to $B$, then by the definition of $B$,  $w(v_{B(z_a)},z_a) \geq w(v_i,z_a)$ for any $z_a$, and hence the weight of $B$ is 
\begin{align*}
\sum_{a \in [m]} w(v_{B(z_a)},z_a)  \geq \sum_{a \in [m]} w(v_i,z_a) \geq W.
\end{align*}
Consequently, overall, the weight of $B$ is at least $\min\{W, nL\}$.
\end{proof}

The second lemma shows that $G_\mathbf{S}$ admits a matching with weight that is relatively close to the social welfare of the optimal matching allocation for any utility profile consistent to the input profile $\mathbf{S}$. We will use this relation in the analysis of the distortion of $f_t$, as well as in the analysis of our randomized matching mechanism later on. 

\begin{lemma}
\label{lem:max-matching-bound}
Let $\osw$ be the social welfare of the optimal matching allocation. There exists a matching of $G_\mathbf{S}$ with weight at least $\delta^{-1}(\osw - n\tau_t)$. 
\end{lemma}

\begin{proof}
For any agent $i$, let $a^*_i$ be the item that $i$ is given in the optimal matching allocation. Clearly, either there exists $j \in [t]$ such that $a_i^* \in S_{i,j}$, or $a_i^* \not\in \bigcup_{j \in [t]} S_{i,j}$. The total utility accumulated by the agents of the second type is at most $n\tau_t$. For the agents of the first type, since $a_i^* \in S_{i,j}$, there is an edge between $v_i$ and $z_{a_i^*}$ in $G_\mathbf{S}$ of weight 
$$\tau_j = \delta^{-j} = \delta^{-1}\cdot \delta^{-j+1} = \delta^{-1}\cdot \tau_{j-1} \geq \delta^{-1} \cdot u_i(a_i^*).$$ 
Hence, the intersection of $A^*$ and $G_\mathbf{S}$ gives us a matching of weight at least $\delta^{-1}(\osw - n\tau_t)$. 
\end{proof}

\begin{theorem}
\label{thm:simple_det_upb}
For $t \in [n]$ and $\delta = \sqrt[t]{2n}$, the distortion of the deterministic matching mechanism $f_t$ is $O(\sqrt[t]{n})$.
\end{theorem}

\begin{proof}
Let $A = f_t(\mathbf{S})$ be the matching computed by the mechanism $f_t$ when given as input an arbitrary input profile $\mathbf{S}$ that is induced by some consistent utility profile $\utp$. First, observe that if $G_\mathbf{S}$ admits a matching of weight $W$ then $\sw(A,\utp) \geq W$. This follows by the fact that if node $v_i$ is matched to node $z_a$ in $G_\mathbf{S}$, then agent $i$ has utility at least $w(v_1,w_a)$ for item $a$.

Second, we argue that the weight of the maximum weight matching of $G_\mathbf{S}$ is at least $\delta^{-1}/2$. 
The total utility of an agent for the items in $\bigcup_{j \in [t]} S_{i,j}$ is at least $1/2$ since the utility for each of the remaining items is at most $\tau_t = \delta^{-t} = 1/(2n)$. Since $\tau_{j-1} = \tau_j \cdot \delta$ with $\tau_0 = 1$, and $v_i(a) \leq \tau_{j-1}$ and $w(v_i,z_a) = \tau_j$  for every $a \in S_{i,j}$, we have
\begin{align*}
    &\frac12 \leq \sum_{j \in [t]} \sum_{a \in S_{i,j}} v_i(a) \leq \sum_{j \in [t]} |S_{i,j}| \tau_{j-1} = \delta \cdot \sum_{j \in [t]} |S_{i,j}| \tau_j \\
    &\implies \sum_{j \in [t]} \sum_{a \in S_{i,j}} w(v_i,z_a) \geq \frac{\delta^{-1}}{2}.
\end{align*}
Since the edge $(v_i,z_a)$ exists in $G_\mathbf{S}$ only if $a \in \bigcup_{j \in [t]} S_{i,j}$, we have that 
$\sum_{a \in \itms} w(v_i,z_a) \geq \delta^{-1}/2$ and $w(v_i,z_a) \geq \tau_t = 1/(2n)$, and thus, by Lemma~\ref{lem:bigraph}, we have that $G_\mathbf{S}$ admits a matching of weight at least $\min\{\delta^{-1}/2,1/2\}$, which is equal to $\delta^{-1}/2$ since $\delta > 1$. 
By Lemma~\ref{lem:max-matching-bound}, since $n\tau_t = 1/2$, we also have that $G_\mathbf{S}$ admits a matching of weight at least $\delta^{-1}(\osw - 1/2)$, where $\osw$ is the optimal social welfare. 
Consequently, we overall have established that the social welfare of the allocation $A$ computed by the mechanism is 
$$\sw(A) \geq \max\bigg\{ \frac12 \cdot \delta^{-1}, \delta^{-1}(\osw - 1/2) \bigg\}.$$
Hence, the distortion is at most
\begin{align*}
\frac{\delta \cdot \osw}{\max\{1/2, \osw - 1/2\}}
\end{align*}
If $\osw \leq 1$, then the distortion is at most $2\delta\osw \leq 2\delta$.
Otherwise, if $\osw \geq 1$, then $\osw - 1/2 \geq \osw/2$ and the distortion is at most $\delta \cdot \frac{\osw}{\osw/2} = 2 \delta$.
In any case, the distortion is at most $2\delta \in O(\sqrt[t]{n}).$
\end{proof}

Due to Theorem~\ref{thm:simple_det_upb}, we can achieve linear distortion using a single threshold level (which is a significant improvement compared to the $O(n^2)$ distortion that can be achieved with ordinal information) and constant distortion using a logarithmic number of threshold levels. 

\begin{corollary}
We can achieve distortion $O(n)$ by using one threshold level and distortion $O(1)$ by using $t = O(\log{n})$ threshold levels. 
\end{corollary}

We now turn our attention to randomization. We design a mechanism that is a convex combination of the naive rule which chooses a random matching equiprobably among all possible ones, and the deterministic matching mechanism $f_t$ that was analyzed above. 

\begin{definition}
The \textit{randomized matching mechanism $R_t$} with probability $1/2$ chooses a matching uniformly at random, and with probability $1/2$ runs the deterministic mechanism $f_t$ with threshold vector $\tauv = (\delta^{-1}, \delta^{-2}, \ldots, \delta^{-t})$ for $t \in [n]$ and some $\delta > 1$.
\end{definition}

\begin{theorem}
\label{thm:simple_rand_upb}
For $t \in [n]$ and $\delta = \sqrt[t+1]{n}$, the distortion of the randomized matching mechanism $R_t$ is $O(\sqrt[t+1]{n})$.
\end{theorem}

\begin{proof}
Let $A^*$ be an optimal matching with social welfare $s^*$, and $\mtch^2$ the matching computed by the second part of the mechanism (the outcome of $f_t$).
In the first part of the mechanism (where a random matching is chosen with probability $1/2$), since each possible matching has probability at least $1/(n!)$, each agent is matched to each item with probability at least $1/n$. Since the sum of the utilities of each agent for all items is $1$, the expected social welfare from the first part is at least 
$$\frac12 \sum_{i \in \ags} \sum_{a \in \itms} \frac{1}{n} \cdot u_i(a) = \frac12 \cdot n \cdot \frac{1}{n} = \frac12.$$ 
For the second part of the mechanism (where the deterministic mechanism $f_t$ using $\tauv$ is employed), 
since $s := n \tau_t = n \delta^{-t} = \sqrt[t+1]{n}$, by Lemma~\ref{lem:max-matching-bound}, there is a matching in $G_\mathbf{S}$ of weight at least $\delta^{-1}(\osw - s)$, and thus the expected social welfare of the mechanism from the second part is at least $\frac12 \cdot \delta^{-1}(\osw - s)$.
Overall, we have established that
\begin{align*}
\E_{\mtch \sim R_t(\mathbf{S})}[\sw(\mtch)] 
&\geq \frac12 + \frac12 \cdot \delta^{-1}(\osw - s) \\
&\geq \frac12 \cdot \max\bigg\{ 1, \delta^{-1}(\osw - s) \bigg\},
\end{align*}
and thus the distortion is at most
\begin{align*}
\frac{2 \cdot \osw}{\max\{1, \delta^{-1}(\osw - s)\}}.
\end{align*}
If $\osw \geq 2s$, then $\osw - s \geq \osw/2$ and the distortion is at most 
$\frac{2 \cdot \osw}{\delta^{-1}\osw/2} = 4\delta = 4\sqrt[t+1]{n}.$ 
Otherwise, if $\osw < 2s$, the distortion is at most 
$2\osw \leq 4s = 4\sqrt[t+1]{n}.$
In any case, the distortion $O(\sqrt[t+1]{n})$. 
\end{proof}

\section{Generalized Setting}
\label{sec:extention}

In this section we consider a generalized setting. Similarly to before, $\ags$ represents a set of $n \geq 1$ agents. However, here it is not necessarily the case that we have an equal number of items; we define $\itms$ to be a set of $m \geq 1$ items. Each item $a \in \itms$ has a supply $m_a \geq 1$, and each agent $i \in \ags$ has a capacity $c_i \geq 1$. For simplicity, we assume that the total supply is equal to the total capacity, that is, $\sz:= \sum_{i \in \ags} c_i = \sum_{a \in \itms} m_a$.\footnote{Our results hold even when $\sum_{i \in \ags} c_i = \Theta(\sum_{a \in \itms} m_a)$, in which case we would need to define $T$ as the maximum between these two quantities.} 

Agents are allowed to receive copies of the same item, in which case their utility depends on the number of copies they receive; in other words, copies of an item are not considered independent. For each agent $i$, item $a$ and $j \in [\min\{c_i, m_a\}]$, we denote by $u_i(a,j)$ the \textit{marginal utility }that agent $i$ gets when receiving his $j$-th copy of item $a \in \itms$, and by $u^+_i(a,j)$ his \textit{total utility} when receiving $j \leq \min\{c_i, m_a\}$ copies of item $a \in \itms$, i.e., $u^+_i(a,j) = \sum_{k \in [j]} u_i(a,j)$.

An {\em allocation} $X = (x_i(a))_{i \in \ags, a \in \itms}$ determines the number $x_i(a)$ of copies of item $a$ that agent $i$ is assigned to, such that $\sum_{a \in \itms} x_i(a) \leq c_i$ for every $i \in \ags$ and $\sum_{i \in \ags} x_i(a) \leq m_a$ for every $a \in \itms$. Given an allocation $X$, the utility of $i$ for $X$ is 
$u_i(X) = \sum_{a \in \itms} u^+_i(a,x_i(a))$. 
We assume that the utility function of each agent $i$ satisfies the unit-sum assumption, that is, 
$\sum_{a \in \itms} u^+_i(a, \min\{c_i, m_a\}) = 1$.\footnote{Our results hold even if we replace the unit-sum assumption by the assumption that the total utilities of the agents (as if each of them is given all items he can hold) are not equal, but known.}
The definition of the social welfare of an allocation is the same as before, that is, it is the total utility of the agents for the allocation. 
We aim to compute allocations with high social welfare that maximally assign the items to the agents; observe that the maximum possible number of items that can be allocated is $\sz$ (or $\min\left\{\sum_{i \in \ags} c_i, \sum_{a\in \itms} m_a\right\}$ in the more general case), and that any allocation $Y$ that assigns less than $\sz$ items is dominated in terms of social welfare by any allocation $X$ that assigns the items allocated by $Y$ in the same way, but also somehow assigns the remaining items. So, in the following, we focus on such maximal allocations only. 

Since the utility functions depend on the number of item copies that the agents receive, we need to appropriately redefine the \textit{elicitation method}. For a threshold vector $\tauv = (\tau_1, \ldots, \tau_t)$, each agent $i$ reports $t$ disjoint threshold approval sets $S_{i,1}, \ldots, S_{i,t}$, where $S_{i,k}$ includes pairs of items and indices for which $i$ has marginal utility in [$\tau_{k-1}, \tau_k)$, where $\tau_0 := 1$. In other words, $S_{i, k} = \{a \in \itms, j \in [\min\{c_i,m_a\}] \colon \tau_{k-1} \geq  u_i(a, j) > \tau_k\}$. The input profile $\mathbf{S}$ now consists of these threshold approval sets reported by all agents. The definition of distortion can also be appropriately refined by taking the worst case over utility profiles and input profiles consistent with them.

We now define a parametric min-cost flow instance that will be used by our algorithms later. 

\begin{definition}[Min-cost Flow Instance]
\label{def:mcf}
Let $C = (C_i)_{i \in \ags}$ be a capacity vector of the $n$ agents with maximum value $c$, 
$M = (M_a)_{a \in \itms}$ be a supply vector of the $m$ items, 
and $V_{n\times m \times c}$ a value matrix in which $V(i,a,j)$ is the value of agent $i$ when receiving $j$ copies of item $a$. 
Let $G(C, M, V)$ be the following min-cost flow instance:
    $G$ has a source node $s$ and a destination node $t$. Furthermore, there is a node $v_i$ for each agent $i$ and a node $z_a$ for each item $a$. For every $i$, there is an edge $(s,v_i)$ with capacity $C_i$ and cost $0$.
    For every $a$, there is an edge $(z_a,t)$ with with capacity $M_a$ and cost $0$.
    For each agent $i$ and item $a$, we add a component to the graph as shown in \Cref{fig:component}. 
    Nodes $s$ and $t$ have supply and demand equal to $\min\left(\sum_{i \in [n]} C_i,\sum_{a \in [m]} M_a\right)$, respectively. The goal is to find the minimum cost for satisfying this flow.
\end{definition}

\begin{figure*}[ht]
    \centering
\begin{tikzpicture}[scale=0.95]
\begin{scope}[every node/.style={inner sep=0,outer sep=1mm,fill, circle, minimum size=3mm},every label/.style={rectangle,fill=none}]
    \node [label={$v_i$}] (Vi) at (0,0) {};
    \node [label={below:$v_a$}](Va) at (12,-3) {};
    \node [label={$v_{i,a,1}$}] (V1) at (2, 0) {};
    \node [label={$v_{i,a,2}$}](V2) at (6,0) {};
    \node [label={$v_{i,a,M_a-1}$}](V3) at (8,0) {};
    \node [label={$v_{i,a,M_a}$}] (V4) at (12,0)  {};
    \node [label={below:$t$}](t) at (14,-3) {} ;
    \node [label={below:$s$}](s) at (-2,-3) {} ;
\end{scope}

\node (D) at (7, 0) {$\ldots$};

\begin{scope}[>={Stealth[black]}]
    \path [->] (s) edge node[above,sloped] {\scriptsize$C_i$} node[below,sloped] {\scriptsize$0$} (Vi);         
    \path [->] (Vi) edge node[above] {\scriptsize$M_a$} node[below] {\scriptsize$-V_{i,a,1}$} (V1); 
    \path [->] (V1) edge node[above] {\scriptsize$M_a-1$} node[below] {\scriptsize$V_{i,a,1}-V_{i,a,2}$} (V2); 
    
    \path [->] (V3) edge node[above] {\scriptsize$1$} node[below] {\scriptsize$V_{i,a,M_a-1}-V_{i,a,M_a}$} (V4); 
    \path [->] (V4) edge node[sloped,midway,above] {\scriptsize$1$} node[sloped,midway,below] {\scriptsize$V_{i,a,M_a-1}-V_{i,a,M-a}$} (Va); 
    \path [->] (Va) edge node[midway,above,sloped] {\scriptsize$M_a$} node[midway,below,sloped] {\scriptsize$0$} (t); 
    \path [->, bend right] (V1) edge node[midway,above,sloped] {\scriptsize$1$} node[midway,below,sloped] {\scriptsize$0$}(Va);
    
    \path [->, bend right] (V2) edge node[midway,above,sloped] {\scriptsize$1$} node[midway,below,sloped] {\scriptsize$0$}(Va);
    \path [->, bend right] (V3) edge node[midway,above,sloped] {\scriptsize$1$} node[midway,below,sloped] {\scriptsize$0$}(Va);
\end{scope}
\end{tikzpicture}
  
    \caption{Graph component between each pair $(i, a) \in \ags \times \itms$.}
    \label{fig:component}
\end{figure*}
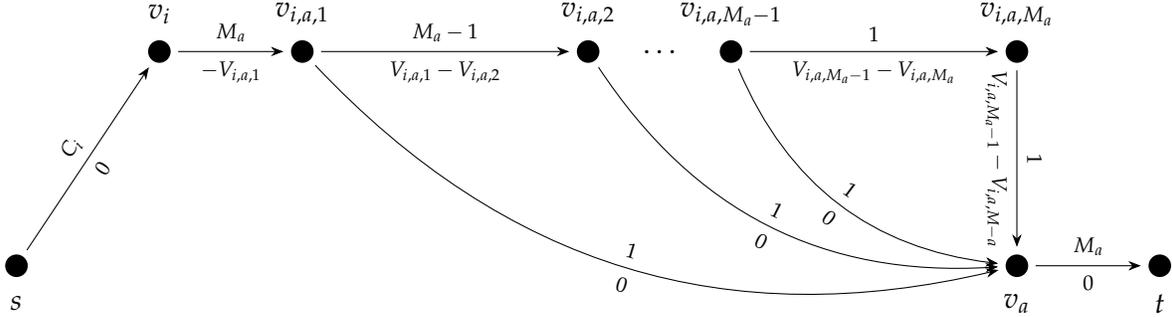

It is well-known that the minimum-cost flow problem can be solved in polynomial time via linear programming (and also using various other algorithms), and we thus have the following property. 

\begin{restatable}{lemma}{integral}
\label{lem:integral}
    For capacity vector $C$, supply vector $M$ and value matrix $V$, the min-cost flow instance defined in \Cref{def:mcf} has an integral solution which we can find in polynomial time.
\end{restatable}

Our next lemma provides a connection between the solution of the min-cost flow instance of Definition~\ref{def:mcf} and the social welfare of the corresponding allocation for the instance of our problem. 

\begin{restatable}{lemma}{graphToSW}
\label{lem:graphToSW}
    The absolute value of the minimum-cost flow in $G(C, M, V)$ is equal to the maximum social welfare of an allocation of items to agents with respect to values in $V$.
\end{restatable}
  
\begin{proof}
Let $X^*$ be the allocation with the maximum social welfare w.r.t. $V$, and $F$ be the min-cost flow in $G(C, M, V)$. Let $F(v, v')$ be the flow from node $v$ to node $v'$, and $\cst(F)$ the cost of $F$. We will show that $\sw(X^*, V) = -\cst(F)$ by bounding $\cst(F)$ from above and below by $-\sw(X^*,V)$.

To prove $\cst(F) \le -\sw(X^*, V)$, we use $X^*$ to construct a flow $F^*$ as follows: 
\begin{itemize}
\item $F^*(v_i, v_{i, a, 1}) = x^*_i(a)$ for $i \in \ags$ and $a \in \itms$;
\item $F^*(v_{i, a, j}, v_{i, a, j+1}) = x^*_i(a) - j$ for $j \in [x^*_i(a)-1]$; 
\item $F^*(v_{i, a, j}, v_a) = 1$ for  $j \in [x^*_i(a)]$;
\item $F^*(s, v_i) = \sum_{ a \in \itms} x^*_i(a)$ for $i \in \ags$;
\item $F^*(v_a, t) = \sum_{ i \in \ags} x^*_i(a)$ for each $a \in \itms$.
\end{itemize}
Observe that these flows satisfy the capacities of the edges, and thus $F^*$ is valid. 
Since $F$ is the cost-minimizing flow, we have
\begin{align*}
\cst(F) &\leq \cst(F^*) \\
        &=\sum_{i \in \ags} \sum_{a \in \itms} -V_{i, a, 1} \cdot F^*(v_i, v_{i, a, 1}) 
         +\sum_{i \in \ags} \sum_{a \in \itms}\sum_{j = 1}^{x^*_i(a)-1} (V_{i, a, j}-V_{i, a, j+1}) \cdot F^*(v_i, v_{i, a, 1})  \\
        &=\sum_{i \in \ags} \sum_{a \in \itms} -V_{i, a, 1} \cdot x^*_i(a)  
         + \sum_{i \in \ags}  \sum_{a \in \itms}\sum_{j = 1}^{x^*_i(a)-1} (V_{i, a, j}-V_{i, a, j+1}) \cdot (x^*_i(a) - j)  \\
        &=\sum_{i \in \ags} \sum_{a \in \itms} -V_{i, a, x^*_i(a)}  
         + \sum_{i \in \ags}  \sum_{a \in \itms}\sum_{j = 1}^{x^*_i(a)-1} V_{i, a, j} ( (x^*_i(a) - j) - (x^*_i(a) - (j-1)) ) \\
        &=\sum_{i \in \ags} \sum_{a \in \itms} -V_{i, a, x^*_i(a)} + \sum_{i \in \ags}  \sum_{a \in \itms}\sum_{j = 1}^{x^*_i(a)-1} -V_{i, a, j} \\
        &=\sum_{i \in \ags}  \sum_{a \in \itms}\sum_{j = 1}^{x^*_i(a)} -V_{i, a, j} = -\sw(X^*, V).
\end{align*}
Next, to prove $\cst(F) \ge -\sw(X^*, V)$, we use $F$ to construct an allocation $X$ with $x_i(a) = F(v_i, v_{i, a, 1})$. Doing similar calculations as above, and since $X^*$ is the welfare-maximizing allocation, we now have that
\begin{align*}
   \cst(F) &= \sum_{i \in \ags}  \sum_{a \in \itms}\sum_{j = 1}^{x_i(a)} -V_{i, a, j} \\
   &= -\sw(X, V) \geq -\sw(X^*, V).
\end{align*}
This completes the proof. 
\end{proof}

We are now ready to present our deterministic mechanism $g_t$, which is a generalization of the deterministic mechanism $f_t$ that we used for the one-sided matching setting. 

\begin{definition}
For $\delta > 1$ and $t \in [n]$, consider the threshold vector $\tauv = (\delta^{-1}, \delta^{-2}, \ldots, \delta^{-t})$. 
The \textit{deterministic generalized matching mechanism} $g_t$ uses the threshold vector $\tauv$ and gets as input a profile $\mathbf{S}$, constant agent capacities $\{c_1, \ldots, c_n\}$, and constant item supplies $\{m_1, \ldots, m_m\}$. 
The mechanism defines the vector
$C = \langle c_1, \ldots, c_n\rangle$, the vector
$M = \langle m_1, \ldots, m_m\rangle$, and the matrix $V$ as follows:
For every $i \in \ags$, $a\in \itms$ and $j \in [\min\{c_i, m_a\}]$, if $(a, j) \in S_{i,k}$ for some $k \in [t]$, then the mechanism defines $V_{i, a, j} = \tau_k$; otherwise, if $(a, j) \not\in \bigcup_{k=1}^t S_{i, k}$, then it defines $V_{i, a, j} = 0$. 
The mechanism computes the solution of the min-cost flow instance defined in \Cref{def:mcf} with input $C$, $M$ and $V$. 
For each agent $i \in \ags$ and item $a \in \itms$ the flow from $v_i$ to $v_{i,a,1}$ in the computed solution is the number of copies of item $a$ that agent $i$ receives.
\end{definition}

Before we bound the distortion of the mechanism, we prove a technical lemma similar to Lemma~\ref{lem:max-matching-bound} which provides us with a lower bound on the social welfare of the allocation computed by the mechanism in relation to the optimal social welfare. 

\begin{lemma} \label{lem:max-matching-bound-generalized}
If there is an allocation with social welfare $\osw$, then $g_t$ outputs an allocation with social welfare at least $\delta^{-1}(\osw - T\cdot \tau_t)$. 
\end{lemma}

\begin{proof}
Let $X^*$ be an optimal allocation. We have
\begin{align*}
\sw(X^*) \leq \sum_{i \in \ags} \sum_{k \in [t]} \sum_{(a, j) \in S_{i, k}} u_i(a, j) + T \cdot \tau_t.
\end{align*}
Recall that, if $(a,j) \in S_{i,k}$ for some $k$, then $ u_i(a, j) \le \delta \cdot V_{i, a, j}$. This implies
$$\sw(X^*) \leq \delta \cdot \sum_{i \in \ags} \sum_{k \in [t]} \sum_{(a, j) \in S_{i, k}} V_{i, a, j} + T \cdot \tau_t.$$ 
Consequently, with respect to $V$, $X^*$ has a social welfare of at least $\delta^{-1} (\sw(X^*) - T \cdot \tau_t)$. 
This is a lower bound on the social welfare of the allocation computed by $g_t$, since, by Lemma~\ref{lem:graphToSW}, this is at least the social welfare of the allocation with maximum social welfare with respect to $V$.
\end{proof}

We are now ready to show the upper bound on the distortion of $g_t$.

\begin{restatable}{theorem}{genDetUP}
\label{thm:gen_det_upb}
For $t \in [\sz]$ and $\delta = \sqrt[t]{2\sz}$, the distortion of the deterministic generalized matching mechanism $g_t$ is $O(c\cdot \sqrt[t]{\sz})$, where $c = \max_{i \in \ags}\{c_i\}$.
\end{restatable}
\begin{proof}
The structure of the proof is very similar to that of Theorem~\ref{thm:simple_det_upb}. 
Let $X = g_t(\mathbf{S})$ be the allocation computed by $g_t$ when given as input an arbitrary input profile $\mathbf{S}$ that is induced by some consistent utility profile $\utp$. Clearly, for any $i \in \ags$, $a\in \itms$, and $j \in [\min\{c_i,m_a\}]$, $u_i(a,j) \geq V_{i, a, j}$ since $V_{i, a, j}$ is defined as a lower bound on this utility. 
By this and Lemma~\ref{lem:graphToSW}, if the graph $G(C,M,V)$ admits a flow of cost $W$, then the social welfare of the corresponding allocation is at least $|W|$. 

We now argue that the mechanism outputs an allocation with social welfare at least $\delta^{-1}/(2c)$. 
To show this, consider any agent $i$. 
Since $\tau_t = \delta^{-t} = 1/(2\sz)$, the sum of the utilities of agent $i$ for the pairs $(a, j)$ that he does not include in any approval set is at most $\sz \times 1/(2\sz) = 1/2$. Moreover, by the definition of $\tauv$, $u_i(a, j) \le \delta \cdot V_{i, a, j}$ for every $(a,j) \in \bigcup_{k =1}^t S_{i,k}$, and thus, for each agent $i$, we have
\begin{equation}\label{eq:Vbound}
\sum_{k \in [t]} \sum_{(a, j) \in S_{i, k}} u_i(a, j) \geq 1/2 \implies \sum_{k \in [t]} \sum_{(a, j) \in S_{i, k}} V_{i,a,j} \geq \delta^{-1}/2.
\end{equation}
Next, we find a partial allocation $B^*$ with social welfare at least $\delta^{-1}/(2c)$. 
To do so, we bundle all copies of each item together and assign them all to a single agent (up to the capacity of the agent). 
In particular, construct a bipartite graph $G'$ with a node $v_i$ for each agent $i \in \ags$ and a node $z_a$ for each item $a \in \itms$. 
There is an edge of weight $\max_{j \in [\min(c_i, m_a)]} V_{i, a, j}$ between $v_i$ and $z_a$. 
Equation~\eqref{eq:Vbound} implies that 
$$\sum_{k \in [t]} \sum_{(a, j) \in S_{i, k}} w(v_i, z_a) \geq {\delta^{-1}\over2c}.$$
Furthermore, for each edge $(v_i, z_a) $ in this bipartite graph, $w(v_i, z_a) > \tau_t = 1/(2T)$. 
Hence, by \Cref{lem:bigraph}, $G'$ has a matching of weight $\min\{\delta^{-1}/(2c), n/(2T)\}$. Since  
$$\frac{n}{2T} = \frac{n}{2 \sum_{i \in \ags}c_i} \geq \frac{n}{2n c} = \frac{1}{2c} \geq \frac{\delta^{-1}}{2c},$$
we finally conclude that assigning all copies of an item to its matched agent in $G'$ gives us a social welfare of at least $\delta^{-1}/(2c)$.

By Lemma~\ref{lem:max-matching-bound-generalized}, since $T\cdot \tau_t = 1/2$, we also have that a lower bound of $\delta^{-1}(\osw - 1/2)$ on the social welfare of the allocation $X$ computed by the mechanism, where $\osw$ is the optimal social welfare. So, overall, we have established that 
$$\sw(X) \geq \max\left\{ \delta^{-1}/(2c), \delta^{-1}(s^* - 1/2) \right\},$$
and the distortion is at most
$$\frac{2\delta \cdot s^*}{\max\left\{ 1/c, 2s^* - 1 \right\}}.$$
If $s^* \geq 1$, then $2s^*-1 \geq s^*$, and hence the distortion is at most $2\delta \in O(\sqrt[t]T)$. 
Otherwise, if $s^* < 1$, then the distortion is at most $2\delta c \in O(c\cdot \sqrt[t]T)$.
\end{proof}

In many applications, the capacities and the supplies are constants. In this case, we have that $T = \Theta(n) = \Theta(m)$ and the following result, which is tight given the corresponding lower bound in Section~\ref{sec:lower}.

\begin{corollary}
When the capacities and supplies are constant, for any $t \in [\sz]$, there is a deterministic mechanism with distortion $O(\sqrt[t]{n})$.
\end{corollary}

We also generalize our randomized mechanism to achieve a slightly better distortion bound. 

\begin{definition}
The \textit{ generalized randomized matching mechanism} $GR_t$ works as follows: 
With probability $1/2$, it bundles the copies of each item together and selects a matching of the items to the agents uniformly at random; once the matching has been chosen, it assigns all copies of an item to its matched agent, subject to capacity and supply constraints. 
With the remaining $1/2$ probability, it runs the deterministic mechanism $g_t$ with threshold vector $\tauv = (\delta^{-1}, \delta^{-2}, \ldots, \delta^{-t})$ for $t \in [n]$ and some $\delta > 1$. 
\end{definition}

The proof of the next theorem follows along the lines of the proof of Theorem~\ref{thm:simple_rand_upb}.

\begin{restatable}{theorem}{genRandUp}
\label{thm:gen_rand_upb}
For $t \in [\sz]$ and $\delta = \sqrt[t+1]{2\sz}$, the distortion of the generalized randomized matching mechanism $GR_t$ is $O(c\cdot\sqrt[t+1]{\sz})$, where $c = \max\{n,m\}/n$. 
\end{restatable}

\begin{proof}
The structure of this proof is similar to that of \Cref{thm:simple_rand_upb}. 
Let $X^*$ be an optimal allocation with social welfare $s^*$, and $X^2$ the allocation computed by the first part of the mechanism (the outcome of $g_t$).

In the first part of the mechanism (where a random matching between agents and items is chosen with probability $1/2$), each agent is matched to each item with probability at least $1/\max\{n,m\}$. Indeed, fix an agent $i$ and an item $a$. If $n\geq m$, there are ${n \choose m} \cdot m!$ different matchings of items to agents, and $i$ gets $a$ in ${n-1 \choose m-1}\cdot (m-1)!$ of them, leading to a probability of 
$$\frac{{n-1 \choose m-1}\cdot (m-1)!}{{n \choose m} \cdot m!} = \frac{{n-1 \choose m-1}}{\frac{n}{m}\cdot {n-1 \choose m-1} \cdot m} = \frac{1}{n}.$$ 
Otherwise, if $m > n$, there are ${m \choose n} \cdot n!$ different matchings, and $i$ gets $a$ in ${m-1 \choose n-1} \cdot (n-1)!$ of them, leading to a probability of 
$$\frac{{m-1 \choose n-1} \cdot (n-1)!}{{m \choose n} \cdot n!} = \frac{{m-1 \choose n-1}}{\frac{m}{n} \cdot {m-1 \choose n-1} \cdot n} = \frac{1}{m}.$$
Since the sum of the utilities of each agent when given all items up to capacity is $1$, the expected social welfare from the first part is at least 
$$\frac12 \sum_{i \in \ags} \sum_{a \in \itms} \frac{1}{\max\{n,m\}} \cdot u_i^+(a,\min\{c_i,m_a\}) = \frac{n}{2\cdot \max\{n,m\}}.$$ 
For the second part (where the deterministic mechanism $g_t$ using $\tauv$ is employed), 
since $s := T \cdot \tau_t = T \cdot \delta^{-t} = \sqrt[t+1]{T}$, by Lemma~\ref{lem:max-matching-bound-generalized}, $g_t$ outputs an allocation with social welfare at least $\delta^{-1}(\osw - s)$, and thus the expected social welfare of the mechanism from the second part is at least 
$\frac12 \cdot \delta^{-1}(\osw - s)$.

Overall, we have established that
\begin{align*}
\E_{X \sim GR_t(\mathbf{S})}[\sw(X)] 
&\geq  \frac{n}{2\cdot \max\{n,m\}} + \frac12 \cdot \delta^{-1}(\osw - s) \\
&\geq \frac12 \cdot \max\bigg\{ \frac{n}{\max\{n,m\}}, \delta^{-1}(\osw - s) \bigg\},
\end{align*}
and thus the distortion is at most
\begin{align*}
\frac{2 \cdot \osw}{\max\left\{\frac{n}{\max\{n,m\}}, \delta^{-1}(\osw - s)\right\}}.
\end{align*}
If $\osw \geq 2s$, then $\osw - s \geq \osw/2$ and the distortion is at most 
$$\frac{2 \cdot \osw}{\delta^{-1}\osw/2} = 4\delta = 4\sqrt[t+1]{T}.$$
Otherwise, if $\osw < 2s$, the distortion is at most 
$$2\cdot \frac{\max\{n,m\}}{n}\cdot \osw \leq \frac{4\max\{n,m\}}{n}\cdot \sqrt[t+1]{T}.$$
In any case, the distortion $O(c \cdot \sqrt[t+1]{T})$, where $c = \max\{n,m\}/n$.
\end{proof}

Again, when the capacities and the supplies are constant, since the total capacity is of the same magnitude as the total supply, it follows that $n$ and $m$ are also of the same magnitude. Hence, the parameter $c = \max\{n,m\}/n$ is a constant and $T=\Theta(n)=\Theta(m)$, giving us the following result, which is again tight due to the corresponding lower bound from Section~\ref{sec:lower}. 

\begin{corollary}
When the capacities and supplies are constant, for any $t \in [\sz]$, there is a randomized mechanism with distortion $O(\sqrt[t+1]{n})$.
\end{corollary}

\begin{remark}
In the generalized setting that we considered in this section, the agents are allowed to receive potentially all available copies of the items, up to their capacity. However, in several applications, we might want to disallow this and set a limit $\ell_{i,a}$ on the number of copies of $a \in \itms$ that agent $i \in \ags$ can get. For example, in the paper assignment problem, each agent must be given at most one copy of each item since it does not make sense for someone to review a paper more than once. Such constraints can be handled in several ways. One of them is via the utility functions of the agents which, for scenarios like these, would simply assign a marginal value of $0$ for any extra copy that exceeds the limit, that is, $u_i(a,j)=0$ for every $j > \ell_{i,a}$. If the utility function is not naturally defined this way, we can modify the min-cost flow instance by setting $V_{i, a, j} = -\infty$ for $j > \ell_{i, a}$, or by removing the corresponding edges in the graph.
\end{remark}

\section{Conclusion and Open Problems} \label{sec:conclusion}
In this paper, we showed tight bounds on the best possible distortion of (both deterministic and randomized) mechanisms for matching settings (that capture important applications, including the one-sided matching problem and the paper assignment problem) when the elicited information about the preferences of the agents is of the form of threshold approvals. Going forward, it would be interesting to explore whether improved tradeoffs can be achieved by using randomization not only for the decision phase of the mechanism but also for the definition of the threshold values, similarly to the works of \citet{benade2021participatory,bhaskar2018truthful}. Furthermore, one could explore other settings in which the same type of elicitation method can be applied, including voting settings (both utilitarian and metric) in which the full potential of using multiple threshold approvals has not been considered before, as well as other resource allocation settings, potentially also in combination with other constraints, such as truthfulness or fairness. 

\bibliographystyle{named}
\bibliography{bib}

\end{document}